\documentclass[journal,comsoc]{IEEEtran}
\usepackage{amsmath,graphics,amssymb,epsfig,subfigure,color,amsthm,cite,float}
\usepackage{array}
\usepackage{multirow}
\usepackage{enumerate}
\usepackage{makecell}
\usepackage{tabularx} 
\usepackage{algorithm,algcompatible}
\algnewcommand\INPUT{\item[\textbf{Input:}]}%
\algnewcommand\OUTPUT{\item[\textbf{Output:}]}%

\newtheorem{thm}{Theorem}
\newtheorem{lem}{Lemma}

\begin{document}

\title{Orthogonal Sparse Superposition Codes for Ultra-Reliable Low-Latency Communications}

\author{Yunseo Nam,~\IEEEmembership{Student Member,~IEEE}, \IEEEauthorblockN{Jeonghun Park},~\IEEEmembership{Member,~IEEE},\\~ \IEEEauthorblockN{Songnam Hong},~\IEEEmembership{Member,~IEEE}, and \IEEEauthorblockN{Namyoon~Lee},~\IEEEmembership{Senior Member,~IEEE}
  \thanks{Y. Nam and N. Lee are with the Department of Electrical Engineering, POSTECH, Pohang, Gyeongbuk 37673, South Korea  (e-mail: \{edwin, nylee\}@postech.ac.kr). J. Park is with the Department of Electrical Engineering, Kyungpook National University, Daegu, South Korea (e-mail:jeonghun.park@knu.ac.kr).  S. Hong is with the Department of Electrical Engineering, Hanyang University, Seoul, South Korea (e-mail: snhong@hanyang.ac.kr).}
	}

\maketitle

\setlength\arraycolsep{2pt}
\makeatletter
\newcommand{\vast}{\bBigg@{3.5}}
\newcommand{\Vast}{\bBigg@{4.5}}
\makeatother
\vspace{-12mm}

\begin{abstract}
This paper presents a new class of sparse superposition codes for low-rates and short-packet communications over the additive white Gaussian noise channel. The new codes are orthogonal sparse superposition (OSS) codes. The key idea of the OSS codes is to construct a codeword as a superposition of sparse sub-codewords whose support sets are mutually non-overlapping. To construct such codewords in a computationally efficient manner,  a successive encoding method is presented. Harnessing the orthogonal property among sub-codewords, a simple yet near-optimal decoding method is proposed, which performs element-wise maximum a posterior decoding with successive support set cancellation.  This decoder is super-fast by a linear decoding complexity in block lengths, far less than the commercially used channel decoders for modern channel codes. The upper bounds for the block error rates (BLERs) are analytically derived for a few-layered OSS codes as a function of block lengths and code rates. It turns out that a single-layered OSS code achieves the ultimate Shannon limit in the power-limited regime, even with the linear complexity decoder.  Via simulations, the proposed OSS codes are shown to perform better than commercially used coded modulation techniques for low-rate and short-latency communication scenarios.


\end{abstract}

\section{Introduction}

\subsection{Motivation}
The code design in the low-rate and the short-block length regime is particularly significant for ultra-reliable low-latency communications (URLLC) to enable Internet-of-Things (IoT) and enhanced Machine-Type Communications (eMTC) services for 5G and beyond \cite{Durisi,Popovski2014,URLLC, URLLCapp,Short,Short2,NBIoT}. Currently, Long-Term Evolution (LTE) systems have used the low-rate codes by concatenating a powerful moderate-rate code (e.g., Turbo and LDPC codes) with a simple repetition code. For instance, the Narrow-Band IoT standard allows up to 2048 repetitions of a turbo code with rate 1/3 to meet the maximum coverage requirement, in which the effective code rate approximately becomes $R=\frac{1}{3\times 2048}\simeq 1.6\times 10^{-4}$ \cite{NBIoT}. Although this simple code construction method provides mediocre performances with a reasonable decoding complexity in practice, the code performance and decoding complexity might not be sufficient for extremely low-rate and low-latency communication scenarios. From a coding theoretical perspective, finding efficient codes in this regime is a significant yet challenging problem. In this paper, we present a new class of sparse regression codes, which is particularly efficient for extremely low-rate and low-latency communications.




\subsection{Related Work}

Sparse superposition codes (SPARC) is a joint modulation and coding technique initially introduced by Joseph and Barron \cite{SSC_2012}. Unlike the traditional coded modulation techniques \cite{Forney1998,BICM}, a codeword of SPARCs is constructed by the direct multiplication of a dictionary matrix and a sparse message vector under a block sparsity constraint. Using the Gaussian random dictionary matrix with independent and identically distributed (IID) entries for encoding and the optimal maximum likelihood (ML) decoding, SPARC is shown to achieve any fixed-rate smaller than the capacity of Gaussian channels as the code length goes to infinity \cite{SSC_2012}. 

Designing low-complex decoding algorithms for SPARCs is of great interest to make the codes feasible in practice \cite{FSSC,Rush,Cho,Barbier,Greig}. The adaptive successive decoding method and its variations have made a significant progress in this direction \cite{FSSC,Cho}. By interpreting the decoding problem of the SPARC with the $L$ sections as a multi-user detection problem Gaussian multiple-access channel with $L$ users under a total sum-power constraint, the idea of adaptive successive decoding is to exploit both the successive interference cancellation at the decoder with a proper power allocation strategy at the encoder. Decoding algorithms using compressed sensing have also received significant attention as alternatives of the computationally-efficient decoders by a deep connection between SPARCs and compressed sensing \cite{CandesRombergTao2006}. In principle, the decoding problem of SPARCs can be interpreted with a lens through a sparse signal recovery problem from noisy measurements under a certain sparsity structure. Exploiting this connection, approximate message passing (AMP) \cite{AMP}, successfully used in the sparse recovery problem, is proposed as a computationally-efficient decoding method of SPARCs \cite{Barbier,Rush}. One key feature of the AMP decoder is that the decoding performance per iteration can be analyzed by the state evolution property \cite{AMP}. Although both low-complexity decoders can decay the error probability with a near-exponential order in the code length as long as a fixed code rate is below the capacity, the finite length performance of the AMP decoder is much better than that of the adaptive successive hard-decision decoder. However, the performance of SPARCs with such low-complexity decoders is limited in both a very low-rate and a short-block length regime, in which a transmitter sends a few tens of information bits to a receiver using a few hundreds of the channel uses, i.e., $C\leq 0.05$.


 The performances SPARCs in the low-rate and short-length regime can be improved by carefully designing their dictionary matrices with a finite size \cite{Candes2008,Calderbank}. Finding the optimal dictionary matrix for given code rates and block lengths are very challenging tasks. To avoid this difficulty, the common approach in designing the dictionary matrix is to exploit well-known orthogonal matrices. For instance, the use of the Hadamard-based dictionary matrix is shown to provide better performances than that of the IID Gaussian random dictionary matrices in a finite-block length regime \cite{Greig}. The quasi-orthogonal sparse superposition code  (QO-SSC) is another example, in which Zadoff-Chu sequences are harnessed to construct a dictionary matrix ensuring the near-orthogonal property. Interestingly, it performs better than polar codes \cite{Polar} under some short-block length regimes. However, constructing such a near-orthogonal dictionary matrix per code rate and block length requires a high computational complexity. In addition, the decoding complexity and latency of the iterative decoder, called belief propagation successive interference cancellation (BPSIC), cannot meet the stringent requirements of no error floor performance in URLLC.





\subsection{Contributions}
In this paper, we consider a communication over a Gaussian channel in a power-limited regime, in which the capacity $C$ approaches zero as the block length goes to infinity. For an efficient communication in this regime, we introduce a new class of SPARCs, called orthogonal sparse superposition (OSS) codes. The key innovation of the code is to construct a codeword as a superposition of sparse sub-codewords whose support sets are mutually non-overlapping. The major contributions are summarized as follows:

 \begin{itemize}
    \item We first present an encoding method called \textit{successive encoding}. The key idea of the successive encoding is to sequentially select the non-zero supports of all sub-message vectors in such a way that they are mutually exclusive. This successive selection guarantees the orthogonality between all sub-codewords.  The proposed encoding scheme can construct codebooks with flexible rates for a given block length by appropriately choosing the code parameters including the number of sub-messages, sparsity levels per sub-message, and the non-zero alphabets of it.  The proposed OSS code has a number of intriguing aspects. Not only classical permutation modulation codes introduced in 1960' \cite{Slepian} and recently introduced index  modulation techniques \cite{Basar2013,Renzo2011,Lee_TSP} can be interpreted as special cases of the OSS code. In addition, we observe that the nominal coding gain of OSS codes are identical with that of the bi-orthogonal codes \cite{Forney1998}.
        
    \item We also propose a low-complexity decoder while achieving a near-optimal decoding performance. Thanks to the orthogonality between sub-codewords, the idea of the proposed decoder is to perform element-wise maximum a posterior decoding with successive support set cancellation (E-MAP-SSC) using the Bayesian principle \cite{BMP}. The proposed decoder is super-fast because it has linear complexity with a block length. In particular, for a two-layered OSS code, the proposed decoder is equivalent to a simple ordered statistics decoder, while achieving the optimal decoding performance.

    \item We analyze the performance of the proposed encoder and decoder. We first derive an analytical expression of the block error rates (BLERs) when using a two-layered OSS code with a simple ordered statistics decoder in terms of relevant system parameters, chiefly the number of information bits and the code block length $N$. This analytical expression is particularly useful when predicting the minimum required signal-to-noise-ratio (SNR) to achieve the BLER below $10^{-9}$, which is very hard to obtain the performance even with simulations. More importantly, we show that a simple OSS code with a linear complexity decoder can achieve the ultimate Shannon limit in the power-limited regime. This result is remarkable because the Shannon limit in this regime is achievable with the classical bi-orthogonal code using a low-complexity decoder with the Hadamard transform, which requires super-linear decoding complexity, i.e., $\mathcal{O}(N\log N)$. Therefore, the decoding complexity can be reduced by a factor of $\log N$ while attaining the optimal performance in the power-limited regime.  


    \item We verify the exactness of our analytical BLER expressions by comparing them with numerical results for various code rates. In addition, we compare the performance of the proposed OSS code with the polar code in terms of both BLERs and the finite blocklength achievable rates for a given BLER \cite{Polyanskiy}. From this comparison, we observe that our OSS codes with a linear complexity decoder outperform polar codes using the successive cancellation decoder (i.e., a super-linear complexity decoder) in the low-rate and short-block length regime.

\end{itemize}

The rest of this paper is organized as follows. First, we introduce the OSS codes and present the successive encoding method in Section II. Then, Section III provides the low-complexity decoding algorithm for OSS codes called E-MAP-SSC. In Section IV, we analyze the performance of the proposed OSS codes. Simulation results are provided in Section V. Finally, Section VI concludes the paper with some discussion and possible extensions.

 \section{Orthogonal Sparse Superposition Coding}
In this section, we present a novel encoding strategy called successive orthogonal encoding to construct orthogonal sparse superposition codes. To provide a better understanding of the codes, we explain the properties and remarks of the codes.

\subsection{Preliminary}
  Before presenting the code construction idea, we introduce some notations and definitions. 
  
  \vspace{0.1cm}
{\bf AWGN channel:} 
 We consider a transmission of codeword over the AWGN channel. Let $N\in \mathbb{Z}^{+}$ be the block length and $R\in\mathbb{R}^{+}$ be the rate of a code. Then, from a codeword ${\bf c}\in\mathbb{R}^N$ in a codebook $\mathcal{C}\in\mathbb{R}^{N\times2^{NR}}$, the received vector ${\bf y}\in\mathbb{R}^N$ is obtained as
\begin{align}
    {\bf y}={\bf c}+{\bf v},
\end{align}
where ${\bf v}\in\mathbb{R}^N$ is a Gaussian noise vector distributed as $\mathcal{N}({\bf 0},\sigma^2{\bf I})$.
  
  \vspace{0.1cm}
{\bf Dictionary and subcodeword:} 
We let ${\bf x}_{\ell}\in\mathbb{R}^N$ be the $\ell$th sparse message vector with sparsity level of $\|{\bf x}_{\ell}\|=K_{\ell}$ for $\ell\in [L]$. We also define an orthogonal dictionary matrix ${\bf U}\in \mathbb{R}^{N}$, i.e., ${\bf U}^{\top}{\bf U}={\bf I}$.  Then, the $\ell$th subcodeword is the multiplication of an orthogonal dictionary matrix and the sparse vector ${\bf c}_{\ell}={\bf U}{\bf x}_{\ell}$ for $\ell\in [L]$. Here, note that all subcodewords share the same orthogonal dictionary matrix.

\vspace{0.1cm}
{\bf Constellation:} We define a set of signal levels for the element in ${\bf x}_{\ell}$. Let $J_{\ell}$ be a non-zero alphabet size of $x_{\ell}$ for $\ell\in [L]$. Then, the signal level set for the non-zero values of the $\ell$th codeword vector is $	\mathcal{A}_{\ell}=\{ a_{\ell,1}, a_{\ell,1},\ldots, a_{\ell,J_{\ell}}\}$. $\mathcal{A}_{\ell}$ are typically chosen from arbitrary pulse amplitude modulation (PAM) signal sets.  By the union of $L$ signal level sets, we define a multi-level modulation set as ${\mathcal A}= \cup_{\ell=1}^L\mathcal{A}_{\ell}$. For simplicity in the decoding process, we assume that the signal level sets for distinct subcodewords are distinct, i.e., $\mathcal{A}_{\ell} \cap \mathcal{A}_{j}=\phi$ for $\ell\neq j \in [L]$.


\begin{figure*}
	\centering 
   \includegraphics[width=0.9\textwidth]{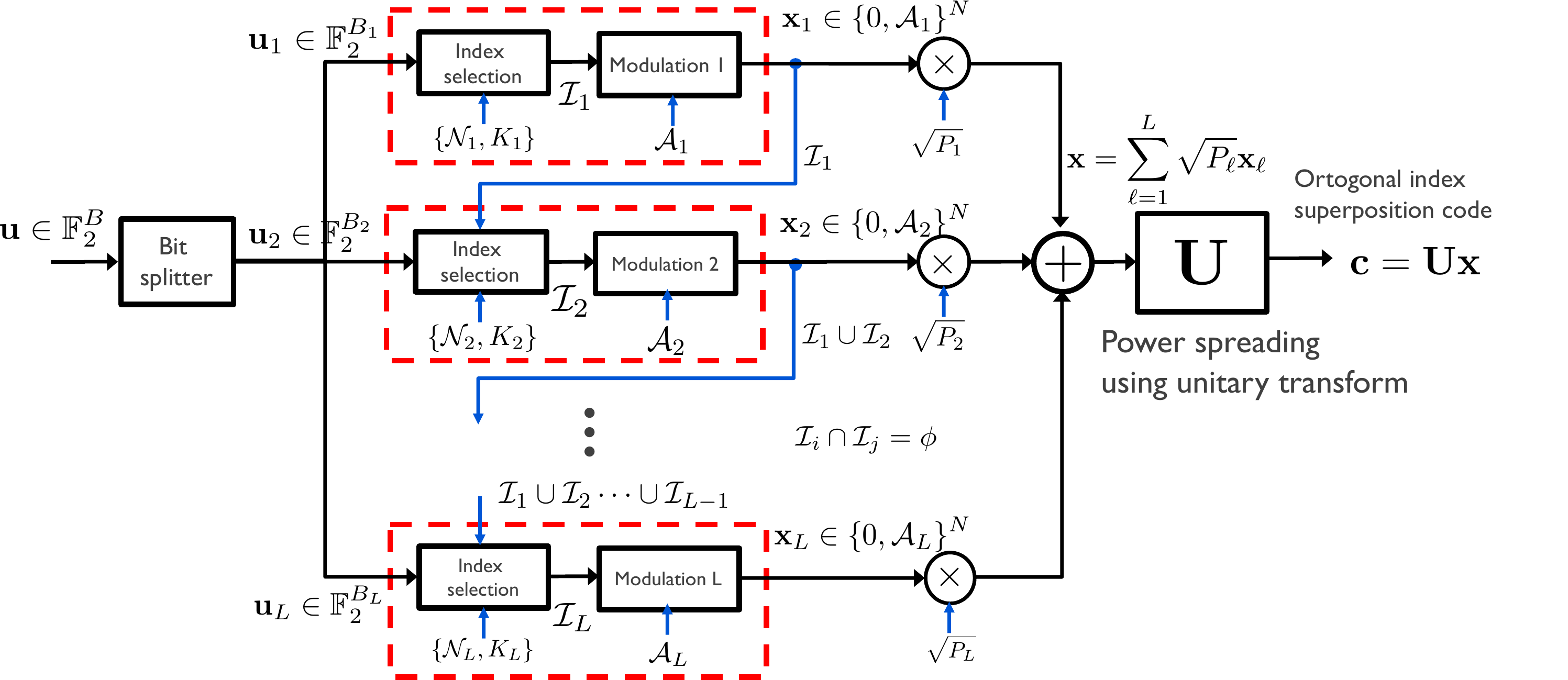}
  \caption{Proposed successive encoder structure for the orthogonal sparse superposition code construction.}  \label{figure:Fig1}
\end{figure*}

\subsection{Successive Encoding}
The idea of successive encoding is illustrated in Fig. \ref{figure:Fig1}. As can be seen, the proposed encoding scheme is to successively map information bits into $L$ orthogonal subcodeword vectors ${\bf c}_{\ell}={\bf U}{\bf x}_{\ell}\in \mathbb{R}^N$, where ${\bf c}_{\ell}^{\top}{\bf c}_{j}=0$ for $\ell\neq j \in [L]$.

We first explain how to construct the first subcodeword ${\bf c}_1$ from the first sparse message vector ${\bf x}_1$.  Let ${\bf u}_{1}\in \mathbb{F}_2^{B_{1}}$ be a binary information strings with length $B_1$. The encoder maps  ${\bf u}_{1}$ into ${\bf x}_{1} \in \left\{ { 0}\cup \mathcal{A}_1\right\}^N$ by uniformly selecting $K_1$ indices in $\mathcal{N}_{1}=[N]$. We define the support set of ${\bf x}_1=[x_{1,1}, x_{1,2}, \ldots ,x_{1,N}]^{\top}$ by $\mathcal{I}_1 =\left\{ n \mid x_{1,n} \in \mathcal{A}_1\right\}$.
Then, the encoder uniformly assigns elements in $\mathcal{A}_{1}$  into the non-zero position of ${\bf x}_{1}$. Therefore, in the first layer, the first subcodeword conveys $B_{1}=\left\lfloor\log_2\left({ |\mathcal{N}_{1}|}\choose{K_1} \right)\right\rfloor+K_1\log_2(|\mathcal{A}_{1}|)$ information bits.

Utilizing the support information of ${\bf x}_1$, i.e., $\mathcal{I}_{1}$, the encoder constructs ${\bf c}_2={\bf U}{\bf x}_2$ so that it is orthogonal to ${\bf c}_1$, i.e., ${\bf c}_2^{\top}{\bf c}_1$=0. To accomplish this, the encoder defines an index set for the second layer as $\mathcal{N}_2\subset [N]/\mathcal{I}_1$, which ensures $\mathcal{I}_{1}\cap \mathcal{I}_{2}\neq \phi$. In the second layer, the encoder maps ${\bf u}_{2}\in \mathbb{F}_2^{B_2}$ into ${\bf x}_{2} \in \left\{ 0, \mathcal{A}_{2} \right\}^N$ by uniformly choosing $K_2$ indices in $\mathcal{N}_{2}$ and uniformly allocating elements in $\mathcal{A}_{2}$ into the non-zero position of ${\bf x}_{2}$. The second codeword ${\bf c}_2={\bf U}{\bf x}_2$, therefore, carries $B_{2}=\left\lfloor\log_2\left({ |\mathcal{N}_{2}|}\choose{K_2} \right)\right\rfloor+K_2\log_2(|\mathcal{A}_{2}|)$ information bits.

The encoder successively applies the same principle until the $L$th layer. Let $\mathcal{N}_L\subset[N]/\left\{\mathcal{I}_1\cup \mathcal{I}_2 \cup \cdots \mathcal{I}_{L-1}\right\}$.  The encoder maps ${\bf u}_{L}\in \mathbb{F}_2^{B_{L}}$   into ${\bf x}_{L}$ by uniformly choosing $K_L$ indices among $\mathcal{N}_L$. Then, it uniformly assigns elements in $\mathcal{A}_L$ in the selected position of ${\bf x}_{L}$. Thanks to the orthogonal construction, the support sets of ${\bf x}_L$, i.e., $\mathcal{I}_{L}$, are mutually exclusive with the union of the support sets of ${\bf x}_{j}$ for $j\in \{1,2,\ldots, L-1\}$, i.e., $\mathcal{I}_{L}\cap \left\{\cup_{j=1}^{L-1}\mathcal{I}_{j}\right\} =\phi$. This guarantees the orthogonality principle among subcodewords, ${\bf c}_{\ell}^{\top}{\bf c}_{j}=0$ for $\ell\neq j \in [L]$.

Finally, a codeword of the orthogonal superposition code is a superposition of $L$ subcodeword vectors, namely, 
\begin{align}
	{\bf c} & = \sum_{\ell=1}^L{\bf U}{\bf x}_{\ell}  = {\bf U}{\bf x} =\sum_{j\in \mathcal{I}}{\bf U}_jx_j,
\end{align}
where ${\bf x}=\sum_{\ell=1}^L{\bf x}_{\ell}$ with $\|{\bf x}\|_0=\sum_{\ell=1}^LK_{\ell}$ and $\mathcal{I} = \{ j \mid  { x}_j \neq 0\} =\cup_{\ell=1}^L\mathcal{I}_{\ell}$.

\subsection{Properties}
To shed further light the significance of our code construction method, we provide some properties of the code. 


\vspace{0.1cm}
{\bf Decodability:} Under noiseless case, an orthogonal sparse superposition code is uniquely decodable since $\mathcal{A}_{\ell}\cap \mathcal{A}_{k}=\phi$ for $\ell\neq k\in [L]$. Suppose ${\bf U}={\bf I}$, i.e., ${\bf c}={\bf x}$. This is true because a decoder distinguishes the $\ell$th subcodewords ${\bf x}_{\ell}$ from ${\bf x} =\sum_{\ell=1}^L{\bf x}_{\ell}$, provided $a_{\ell,i}\notin \cup_{j\neq \ell}\mathcal{A}_j$. Then, the decoder performs an inverse mapping from ${\bf x}_{\ell}$ to ${\bf u}_{\ell}$ to obtain $B_{\ell}$ information bits.

 \vspace{0.1cm}
 {\bf  Orthogonality:} The most prominent property of the orthogonal sparse superposition code is the orthogonality between subcodewords. This property facilitates to perform decoding in a computationally efficient manner, which will be explained in Section III.  

 \vspace{0.1cm}
{\bf Code rate:} The $\ell$th subcodeword conveys $B_{\ell}=\left\lfloor\log_2\left({ |\mathcal{N}_{\ell}|}\choose{K_{\ell}} \right)\right\rfloor+K_{\ell}\log_2(|\mathcal{A}_{\ell}|)$ information bits using $N$ channel uses. Therefore, the rate of the orthogonal sparse superposition code is 
\begin{align}
		R=\frac{\sum_{\ell=1}^L\left\lfloor\log_2\left({ |\mathcal{N}_{\ell}|}\choose{K_{\ell}} \right)\right\rfloor+K_{\ell}\log_2(|\mathcal{A}_{\ell}|)}{N}.\label{eq:rate}
		 \end{align}
For a symmetric case in which $|\mathcal{N}_{\ell}|=M (< N)$, $K_{\ell}=K$, and $|\mathcal{A}_{\ell}|=1$ for $\ell\in [L]$, the code rate simplifies to  $
		R=\frac{L\left\lfloor\log_2\left({ M}\choose{K} \right)\right\rfloor}{N}.$
For a fixed block length $N$, the proposed encoding scheme is able to construct codes with very flexible rates by appropriately choosing the code parameters including the number of layers $L$, the number of non-zero values per layer $K_{\ell}$, the index set size per layer $|\mathcal{N}_{\ell}|$, and the non-zero alphabets in each layer $\mathcal{A}_{\ell}$. These code parameters can be optimized to control the trade-off between code rates and decoding errors. 


 \vspace{0.1cm}   
 {\bf The average transmit power:} One intriguing property of the codes is that its average transmit power is extremely low. Without loss of generality, we set ${\bf U}={\bf I}$, i.e., ${\bf c}={\bf x}$. Since the $\ell$th subcodeword has the sparsity level of $K_{\ell}$ and its non-zero value is chosen from $\mathcal{A}_{\ell}$, the average power of ${\bf x}_{\ell}$ is 
 \begin{align}
 \mathbb{E}\left[\|{\bf x}_{\ell}\|_2^2\right] = \frac{ K_{\ell} \frac{\sum_{i=1}^{|\mathcal{A}_{\ell}|} a_{\ell,i}^2}{|\mathcal{A}_{\ell}|}}{N}.
 \end{align}
 Since all subcodeword vectors are orthogonal, the average power of ${\bf x}$ becomes
 \begin{align}
 	E_s= \mathbb{E}\left[\|{\bf x}\|_2^2\right]=\sum_{\ell=1}^L\mathbb{E}\left[\|{\bf x}_{\ell}\|_2^2\right]=\sum_{\ell=1}^L\frac{ K_{\ell} \sum_{i=1}^{|\mathcal{A}_{\ell}|} a_{\ell,i}^2 }{|\mathcal{A}_{\ell}|N}.   \label{eq:Es}
 \end{align}
 
  \vspace{0.1cm}  
 {\bf Nominal coding gain:} For a codebook $\mathcal{C}=\left\{{\bf c}^1,{\bf c}^2,\ldots, {\bf c}^{2^{NR}}\right\}$ with size of $2^{NR}$, the minimum Euclidean distance of the codebook is defined as
 \begin{align}
 	d_{\rm min} (\mathcal{C})&=\min_{i,j\in [2^{NR}] } \| {\bf c}^{i}-{\bf c}^j\|_2\nonumber\\
 	&=\min_{\{\ell_1,i_1\}\neq\{\ell_2,i_2\}} \|a_{\ell_1,i_1}-a_{\ell_2,i_2}\|_2. \label{eq:dmin}
 	 \end{align} 
Using the definitions in \eqref{eq:rate}, \eqref{eq:Es}, and \eqref{eq:dmin}, the nominal coding gain\cite{Forney1998} of the proposed OSS codes is given by
 \begin{align}
 	\gamma_c(\mathcal{C}) &=\frac{d^2_{\rm min}(\mathcal{C})/4}{E_s/R}. \label{eq:codinggain}
 \end{align}


\vspace{0.1cm}
{\bf Example:} For ease of exposition, we restrict our attention to the case where $[N,L]=[48,2]$ and $K_{1}=K_2=2$ and ${\bf U}={\bf I}$. We also consider two PAM set $\mathcal{A}_1=\{-1,1\}$ and $\mathcal{A}_2=\{-2,2\}$. In this example, we construct an OSS code with rate $R=\frac{1}{2}$ and block length $48$. Using the idea of successive encoding, the encoder generates two subcodeword vectors. By choosing two non-zero positions in $\mathcal{N}_1=\{1,2,\ldots, 48\}$ and allocating them to 1 or -1 uniformly, we map $B_1=\left\lfloor\log_2\left({{48}\choose{2}}\right)2^2\right\rfloor=12$ information bits to ${\bf x}_1$. Without loss of generality, we assume that $\mathcal{I}_1=\{1,2\}$. Then, the encoder maps $B_2=\left\lfloor\log_2\left({{46}\choose{2}}2^2\right)\right\rfloor=12$ bits to ${\bf x}_2$ by selecting two non-zero indices from $\mathcal{N}_2=[N]/\mathcal{I}_1$ with $|\mathcal{N}_2|=46$. It uniformly allocates $2$ or -2 to the non-zero elements in ${\bf x}_2$.  Since each subcodeword has ternary alphabets, the OSS code becomes $ {\bf c}={\bf x}_{1}+  {\bf x}_{2} \in \{-2,-1,0,1,2\}^{48}$.
This code has alphabet size of five and the codeword is sparse, i.e., $\|{\bf c}\|_0=4$. The normalized average transmit power per channel use becomes
\begin{align}
	E_s=\frac{\mathbb{E}\left[\|{\bf x}\|_2^2\right]}{N} = \frac{8}{48} +\frac{2}{48} = \frac{5}{24} .
\end{align}
Since the minimum distance is two, the nominal coding gain of this code becomes
\begin{align}
	  \gamma_c(\mathcal{C})=\frac{1}{5/24}=4.8~(6.8~{\rm dB}).
\end{align}

\subsection{Remarks}
We provide some remarks to highlight the difference with the existing coding and modulation methods.

\vspace{0.1cm}
{\bf Remark 1 (Difference with SPARCs):}  To highlight the difference, we compare the restricted isometry property (RIP) constants of the dictionary matrices of the two codes, which are key metrics to measure the sparse signal recovery performance \cite{RIP}.  For ease of exposition, we consider a two-layer orthogonal sparse superposition code with $\mathcal{A}_1=\{1\}$ and $\mathcal{A}_2=\{-1\}$. In this case, we can define an extended dictionary matrix as
\begin{align}
    \bar{\bf U}&=[1 -1]\otimes{\bf U}\in\mathbb{R}^{N\times 2N},
\end{align}
where $\otimes$ is the Kronecker product. We also define an extended binary sparse vector ${\bf \bar x}\in \{0,1\}^{2N\times 1}$ that satisfies $\bar{\bf U}\bar{\bf x}={\bf U}{\bf x}$ using a one-to-one mapping function $g({\bf x}):\{1,0,-1\}^N \rightarrow \{0,1\}^{2N}$. Consequently, the decoding problem is to recover a binary vector under the block sparsity constraint on ${\bf \bar x}$. By the orthogonal construction, the RIP constant of the dictionary matrix $\bar{\bf U}$ becomes zero, i.e., $\delta=0$. Whereas, the Gaussian random dictionary matrix for SPARCs  has a non-zero RIP constant $\delta>0$, and this constant tends to increase as the dictionary matrix size shrinks. Therefore, the propose OSS code is more beneficial than SPARCs in a low-rate and short block  length regime.

\vspace{0.1cm}
 {\bf Remark 2 (Orthogonal multiplexing for multi-layer index modulated signals):} The proposed coding scheme also generalizes the existing index modulation methods \cite{Basar2013}. Suppose a single-layer encoding with one section, i.e., $L=1$. This method is identical to the index modulation. Therefore, our coding scheme can be interpreted as an efficient orthogonal multiplexing method of a multi-layer index (spatial) modulated signals \cite{Basar2013,Renzo2011,Lee_TSP}. Therefore, one can use the proposed OSS codes as a modulation technique in conjunction with modern codes (e,g., polar and LDPC codes).


 \vspace{0.1cm}
  {\bf Remark 3 (Difference with permutation modulation codes):} One interesting connection is that the classical permutation modulation codes are special cases of our OSS codes with the joint information bit mapping technique. Specifically, by setting signal level sets to be a singleton $\mathcal{A}_{\ell}=\{a_{\ell,1}\}$ for $\ell\in[L]$, it is possible to generate the identical codebook with  Variant I codebook in \cite{Slepian} with a joint mapping, which achieves the rate of
  \begin{equation*}
  R=\frac{\left\lfloor\log_2\left( {\prod_{\ell=1}^L{{N-(\ell-1)K}\choose{K}}}\right)\right\rfloor}{N}.
  \end{equation*} Thanks to the degrees of freedom to design the signal levels per subcodewords and the additional information bit mapping of the non-zero elements, our encoding scheme is able to generate a large codebook for given $N$ and $K$. Besides, the superposition encoding method of multiple subcodewords facilitates to implement the encoder in practice because it dwindles the encoding complexity by the separate bit-mapping technique.

  \section{A Low-Complexity Decoding Algorithm}

 This section presents a low complexity decoding algorithm, referred to as element-wise maximum a posterior decoding with successive support set cancellation (E-MAP-SSC). For simplicity, we assume ${\bf U}={\bf I}$ throughout this section. 
   
 The main idea of the proposed algorithm is to successively decode sparse subcodeword vector ${\bf x}_{\ell}$ in ${\bf x}$ so as to maximize a posterior probability (APP) using the Bayesian approach \cite{MAPsupp,BMP}. Recall that the joint APP can be factorized as
 \begin{align}
    {\sf P}\left( {\bf x} |  {\bf y}\right)&=\prod_{{\ell}=1}^{L}{\sf P}\left({\bf  x}_{\ell}|  {\bf y}, {\bf x}_{\ell-1},\ldots, {\bf x}_2, {\bf x}_1 \right)\nonumber \\
    &=\prod_{{\ell}=1}^{L}{\sf P}\left({\bf  x}_{\ell}|  {\bf y}, {\mathcal{I}}_{\ell-1},\ldots, {\bf \mathcal{I}}_2, {\bf \mathcal{I}}_1 \right),
    \label{eq:product_form}
\end{align}
where the equality follows from the fact that ${\mathcal{I}}_{\ell-1},\ldots, {\bf \mathcal{I}}_2, {\bf \mathcal{I}}_1$ are sufficient information to decode ${\bf x}_{\ell}$. From this decomposition, the proposed E-MAP-SSC aims to successively estimate each subcodeword vector by exploiting the knowledge of previously identified support sets. This motivates us to consider a decoding method with successive support set cancellation.  

The proposed decoder performs $L$ iterations. Each iteration decodes a subcodeword vector and subtracts the contribution of the previously identified support sets to evolve the prior distribution for the next iteration. Suppose the decoder has correctly identified the non-zero support sets ${\mathcal{\hat I}}_{\ell-1},\ldots, {\bf \mathcal{\hat I}}_2, {\bf \mathcal{\hat I}}_1$ in the previous $\ell-1$ iterations. Using this information, it performs element-wise MAP decoding to identify support sets $\mathcal{\hat I}_{\ell}$. Let $\mathcal{\hat N}_{\ell} =[N]/ \cup_{j=1}^{\ell-1}{\mathcal{\hat I}}_{j}$. Since all received signals $y_{n}$ for $n\in \mathcal{\hat N}_{\ell}$ are conditionally independent for given ${\bf x}_{\ell}$, the joint APP is factorized as 
\begin{align}
	&{\sf P}\left( {\bf x}_{\ell} \mid {\bf y}, {\mathcal{\hat I}}_{\ell-1},\ldots, {\bf \mathcal{\hat I}}_2, {\bf \mathcal{\hat I}}_1\right) \nonumber\\
	\!\!&\!\!=\!\frac{1}{Z}\!\prod_{n\in \mathcal{\hat N}_{\ell}}\frac{ {\sf P}\left( y_n \mid { x}_{\ell,n} \right){\sf P}\left(  { x}_{\ell,n} \right)  }{{\sf P}\left( y_n \right)} {\bf 1}_{\left\{ \sum_{n\in \mathcal{\hat N}_{\ell}}{\bf 1}_{\{x_{\ell,n}\in\mathcal{A}_{\ell} \}} =K_{\ell} \right\} }, \label{eq:posteri}
\end{align}
To compute this, we need the likelihood function, which is given by
\begin{align}
	{\sf P}\left( y_n \mid {x}_{\ell,m} \in \mathcal{A}_{\ell} \right) =\frac{1}{|\mathcal{A}_{\ell}|}\sum_{j=1}^{|\mathcal{A}_{\ell}|} \frac{1}{\sqrt{2\pi \sigma^2}}\exp\left(-\frac{|y_n- {a}_{\ell,j}|^2}{2\sigma^2}\right). \label{eq:likelihood}
\end{align}
The prior distribution of ${\bf x}_{\ell}$ is also decomposed into
 \begin{align}
 	{\sf P}\left({\bf x}_{\ell}\right) = \frac{1}{Z}\prod_{n\in \mathcal{\hat N}_{\ell}}{\sf P}(x_{\ell,n}) {\bf 1}_{\left\{ \sum_{n\in \mathcal{\hat N}_{\ell}}{\bf 1}_{\{x_{\ell,n}\in\mathcal{A}_{\ell} \}} =K_{\ell} \right\} },
 \end{align}
where $Z\in \mathbb{R}^{+}$ denotes a constant to be a probability distribution and ${\bf 1}_{\mathcal{C}}$ is an indicator function for set $\mathcal{C}$. Recall that the non-zero supports of ${\bf x}_{\ell}$ are uniformly selected from $\mathcal{\hat N}_{\ell}$ with $|\mathcal{\hat N}_{\ell}|= N-\sum_{j=1}^{\ell-1} K_j$. The probability mass function of ${x}_j$ becomes
 \begin{align}
	{\sf P}({x}_{\ell,n}) =  \begin{cases}
    x_{\ell,n}=0 &  {\rm w.p.}~~~1-\frac{K_{\ell}}{N-\sum_{j=1}^{\ell-1} K_j}\\
    x_{\ell,n}\in  \mathcal{A}_{\ell} & {\rm w.p.}~~~\frac{K_{\ell}}{N-\sum_{j=1}^{\ell-1} K_j}. \label{eq:prior_app}\\
     \end{cases}
\end{align} 
 Invoking the prior distribution of ${x}_{\ell,n}$ in \eqref{eq:prior_app}, we obtain 
\begin{align}
	{\sf P}\left( y_n \right)&= {\sf P}\left( y_n \mid {x}_{\ell,n} \in \mathcal{A}_{\ell} \right){\sf P}\left( {x}_{\ell,n} \in \mathcal{A}_{\ell} \right) \nonumber\\
	&~~+ {\sf P}\left( y_n \mid {x}_{\ell,n} \notin \mathcal{A}_{\ell}  \right){\sf P}\left( {x}_{\ell,n} \notin \mathcal{A}_{\ell}  \right) \nonumber\\
	&=\frac{1}{\sqrt{2\pi \sigma^2}} \sum_{j=1}^{|\mathcal{A}_{\ell}|} e^{-\frac{|y_n- a_{\ell,j}|^2}{2\sigma^2}}\frac{K_{\ell}}{N-\sum_{j=1}^{\ell-1} K_j} \frac{1}{|\mathcal{A}_{\ell}|} \nonumber \\
	&+\frac{1}{\sqrt{2\pi \sigma^2}} e^{-\frac{|y_n|^2}{2\sigma^2}}\left(1-\frac{K_{\ell}}{N-\sum_{j=1}^{\ell-1} K_j}\right) .  \label{eq:dist_y_iter1}
\end{align}
Utilizing \eqref{eq:likelihood}, \eqref{eq:prior_app}, and \eqref{eq:dist_y_iter1},  the decoder computes the probability of event that $n\in \mathcal{I}_{\ell}$ given $y_n$  as
 \begin{align}
&	{\sf P}(n\in \mathcal{I}_{\ell} | y_n)=\frac{ {\sf P}\left( y_n \mid {x}_{\ell,n}\in \mathcal{A}_{\ell} \right){\sf P}\left(  {x}_{\ell,n}\in \mathcal{A}_{\ell} \right)  }{{\sf P}\left( y_n \right)}  \nonumber \\
&=  \frac{  \sum_{j=1}^{|\mathcal{A}_{\ell}|} e^{-\frac{|y_n- a_{\ell,j}|^2}{2\sigma^2}}\frac{K_{\ell}}{N-\sum_{j=1}^{\ell-1} K_j} \frac{1}{|\mathcal{A}_{\ell}|} }{  \sum_{j=1}^{|\mathcal{A}_{\ell}|} e^{-\frac{|y_n- a_{\ell,j}|^2}{2\sigma^2}}\frac{K_{\ell}}{N-\sum_{j=1}^{\ell-1} K_j} \frac{1}{|\mathcal{A}_{\ell}|} + e^{-\frac{|y_n|^2}{2\sigma^2}}\left(1-\frac{K_{\ell}}{N\sum_{j=1}^{\ell-1} K_j}\right)}.\label{eq:APP}
\end{align}
To satisfy the sparsity condition in ${\bf x}_{\ell}$, i.e., $\sum_{n\in \mathcal{\hat N}_{\ell}}{\bf 1}_{\{x_{\ell,n}\in\mathcal{A}_{\ell} \}} =K_{\ell}$, the decoder estimates the support sets of ${\bf x}_{\ell}$ by selecting $K_{\ell}$ indices that provide the $K_{\ell}$-largest element-wise MAP probabilities \eqref{eq:APP}. Let ${\hat i}_{\ell,k}$ be the ordered index that has the $k$th largest element-wise MAP, i.e.,  ${\sf P}\left({\hat i}_{\ell,1} \in \mathcal{I}_{\ell} | y_{\ell,1}\right)>{\sf P}\left({\hat i}_{\ell,2} \in \mathcal{I}_{\ell} | y_{\ell,2}\right),\ldots$ for ${\hat i}_{\ell,k}\in \mathcal{\hat N}_{\ell}$.  Then, the estimated support set of ${\bf x}_{\ell}$ is
\begin{align}
	\mathcal{\hat I}_{\ell}=\left\{ {\hat i}_{\ell,1}, {\hat i}_{\ell,2},\ldots, {\hat i}_{\ell,K_{\ell}} \right\}.
\end{align}
Once the support set is identified in the first step, the decoder performs MAP estimate for the signal levels for $x_{\ell,n}$ for $n\in \mathcal{\hat I}_{\ell}$ in the second step. In particular, this simplifies to the minimum Euclidean distance decoding given by
\begin{align}
		 {\hat x}_{\ell,n} &= \arg\max_{ a_{\ell,j}\in \mathcal{A}_{\ell}} {\sf P}\left(x_{\ell,n}=a_{\ell,j} | y_n,  n\in \mathcal{\hat I}_{\ell}\right) \nonumber\\
		 &= \arg\min_{a_{\ell,j}\in \mathcal{A}_{\ell}}  |y_n- a_{\ell,j}|^2.
		 \end{align}
By repeatedly performing these procedures, the iteration ends when $\ell=L$. Our decoding algorithm is illustrated in Fig. \ref{figure:Fig7}.

\begin{figure*}
	\centering 
   \includegraphics[width=0.7\textwidth]{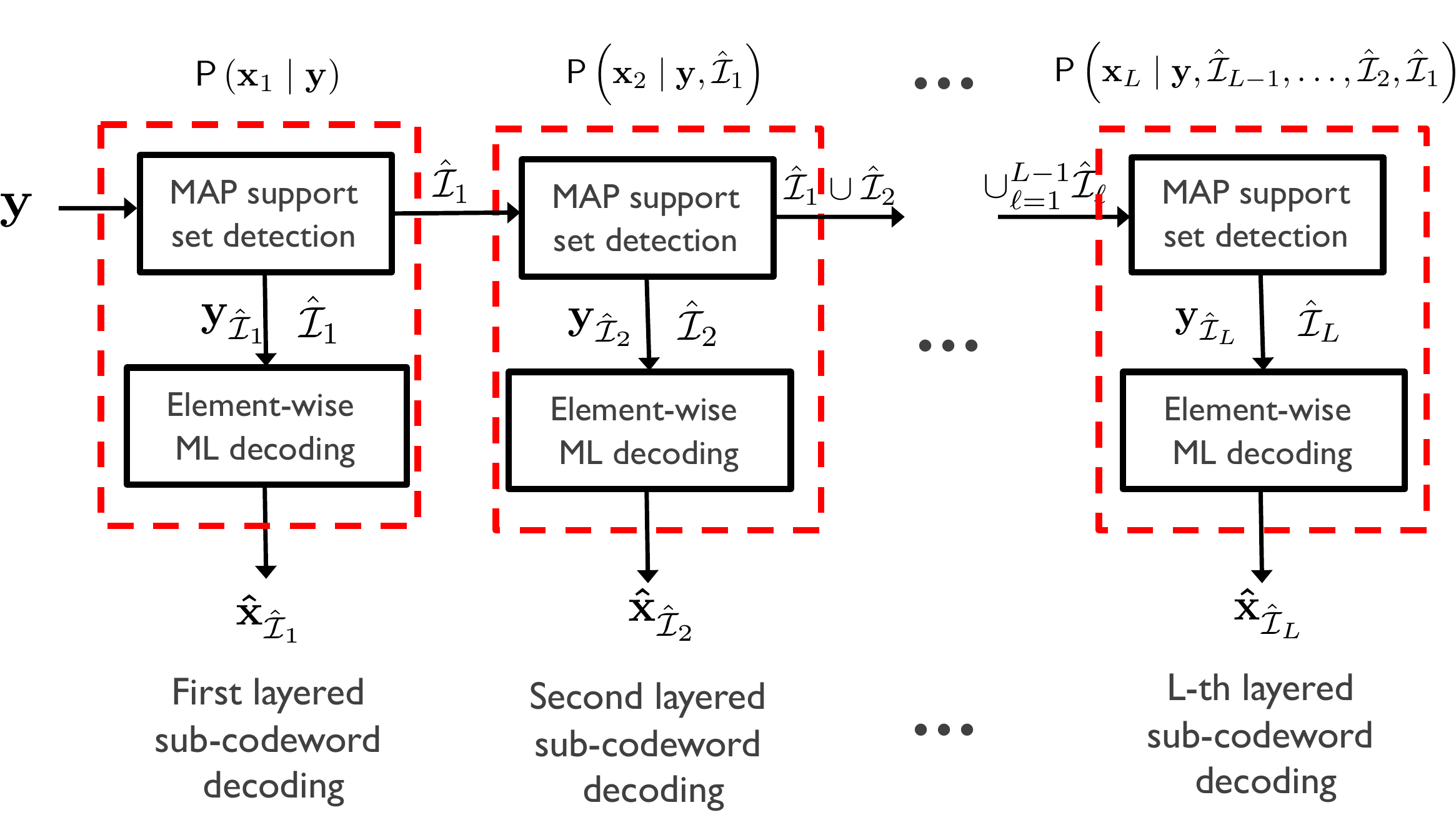}
  \caption{Proposed E-MAP-SSC decoder.}  \label{figure:Fig7}
\end{figure*}

\vspace{0.1cm}
{\bf Remark 4 (Linear decoding complexity):} The decoding complexity of the proposed E-MAP-SSC is linear in block length $N$. Suppose the $\ell$th iteration. In this iteration, the decoder computes $N-\sum_{j=1}^{\ell-1}K_j$ APPs in an element-wise manner as in \eqref{eq:APP}. Then, it identifies the support set $\mathcal{\hat I}_{\ell}$ by selecting the $K_{\ell}$ largest indices among $N-\sum_{j=1}^{\ell-1}K_j$ APP values.  The computational complexity order in finding $\mathcal{\hat I}_{\ell}$ is $\mathcal{O}\left( \left( N-\sum_{j=1}^{\ell-1}K_j \right)\log (K_{\ell})\right)$.  Once $\mathcal{\hat I}_{\ell}$ is obtained, the decoder performs maximum likelihood signal level detection for the $\ell$th layer, which takes $|\mathcal{A}_{\ell}|{K_{\ell}}$ computations.  As a result, the total decoding complexity becomes $\sum_{\ell=1}^L\left(N-\sum_{j=1}^{\ell-1}K_j\right)\log(K_{\ell}) + |\mathcal{A}_{\ell}|{K_{\ell}}$. Under the premise that $N\gg K_{\ell}$ for $\ell\in[L]$, it roughy becomes $\mathcal{O}\left(LN\right)$, which is linear in both the block length and the number of layers.

%
%
%

\vspace{0.1cm}
{\bf Remark 5 (Optimality of a simple ordered statistics decoder):} We consider a two-layer OSS code with $K_1$, $K_2$, $\mathcal{A}_1=\{1\}$, and $\mathcal{A}_2=\{-1\}$. In this case, we show that the proposed E-MAP-SSC algorithm is equivalent to a simple ordered statistics decoder. To see this, we compute the posterior probability of an event $n\in\mathcal{I}_1$. By plugging 
\begin{align}
    &{\sf P}\left( y_n \mid n\in \mathcal{I}_1\right)=\exp\left(-\frac{\left(y_n-1\right)^2}{2\sigma^2}\right),\nonumber\\
    &{\sf P}\left( y_n \mid n\notin \mathcal{I}_1\right)=\frac{(N-K_1)\exp\left(-\frac{y^2_n}{2\sigma^2}\right)+K_1\exp\left(-\frac{(y_n+1)^2}{2\sigma^2}\right)}{N-2K_1}, \nonumber\\
    &{\sf P}\left( n \in \mathcal{I}_{1} \right)=\frac{K_1}{N}, \ \ \text{and} \ \ {\sf P}\left( n \notin \mathcal{I}_{1} \right)=\frac{N-2K_1}{N}
\end{align}
into \eqref{eq:APP}, we obtain
\begin{align}
    	{\sf P}(n\in \mathcal{I}_1 | y_n)&=\frac{1}{1+\exp\left(-\frac{2y_n}{\sigma^2}\right)+\frac{\left(N-2K_1\right)}{K_1}\exp\left(-\frac{2y_n-1}{2\sigma^2}\right)}.
\end{align}
Since ${\sf P}(n\in \mathcal{I}_1 | y_n)$ is a monotonically increasing function of $y_n$, we conclude that the index set $\mathcal{\hat I}_{1}$ is determined by the $K_1$ largest values in $y_n$. Similarly, $\mathcal{\hat I}_{2}$ is determined by the $K_2$ smallest values in $y_n$. This simple decoding method is illustrated in Fig. \ref{figure:Fig8}.

\begin{figure}
	\centering 
   \includegraphics[width=0.5\textwidth]{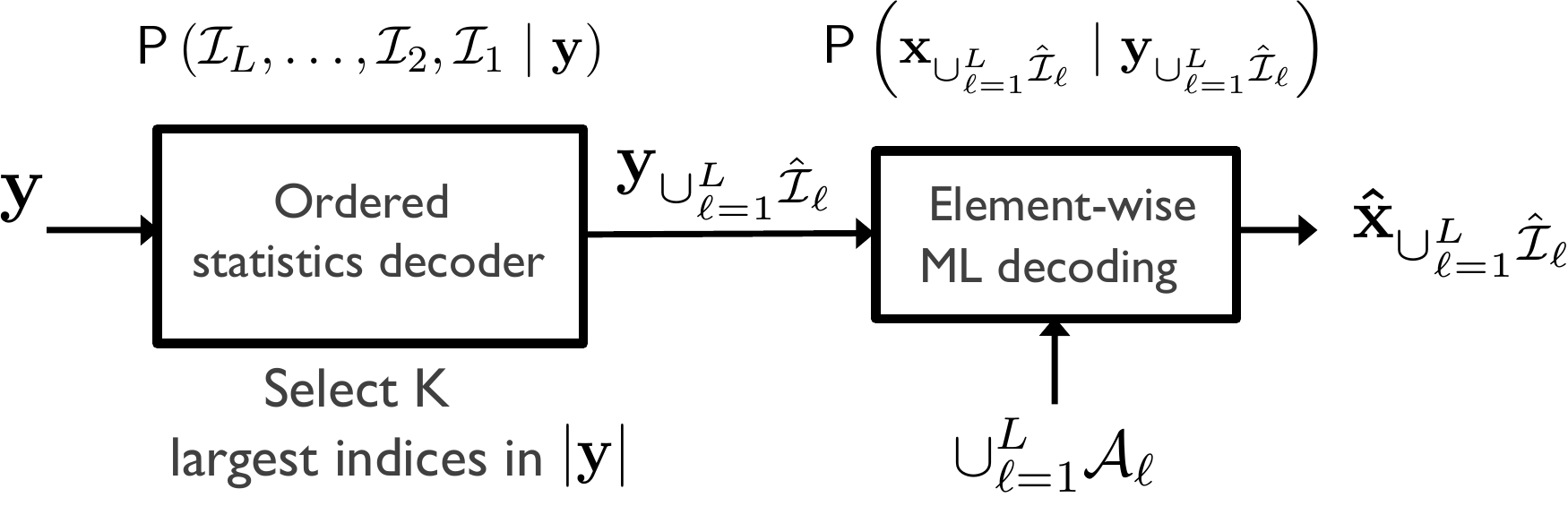}
  \caption{Simple ordered statistics decoder.}  \label{figure:Fig8}
\end{figure}

\section{Performance Analysis}
In this section, we first provide the analytical expressions of the BLERs of simple OSS codes in terms of relevant code parameters. We also show that the simple code achieves the ultimate Shannon limit in the power-limited regime.  We analyze the coding gains of our code to gauge the performance of a finite block length.  

\subsection{BLER Performance Analysis}

The following theorem shows an exact expression of BLERs for a single-layered OSS code.\vspace{0.1cm}
\begin{thm} \label{thm1}
For a single-layered OSS code with rate $R=\frac{ \left\lfloor\log_2\left({ N}\choose{K_{1}} \right)\right\rfloor}{N}$, the BLER is \begin{align}  
        {\sf P}({\mathcal E}) \!=\! 1\! -\! \frac{(N\!-\!K_1)}{\sqrt{2\pi\sigma^2}} \!\!\int_{-\infty}^{\infty} \!\!\! Q\left(\frac{y\!-\!1}{\sigma}\right)^{K_1}\!\!\!\left\{1\!-\!Q\left(\frac{y}{\sigma}\right)\right\}^{\!N\!-\!K_1\!-\!1}\!e^{-\frac{y^2}{2\sigma^2}}{\rm d}y,\label{eq:thm1}
\end{align}
where $Q(x)=\int_{x}^{\infty} \frac{1}{\sqrt{2\pi}} e^{-\frac{u^2}{2}}{\rm d}u$.
\end{thm}

\begin{proof}
See Appendix A.
\end{proof}


The following theorem provides a tight upper bound of the BLER for two-layer OSS codes.

\begin{thm} \label{thm2}
For a two-layer OSS code with $R=\frac{ \left\lfloor\log_2\left({ N}\choose{K} \right)\right\rfloor + \left\lfloor\log_2\left({ N-K}\choose{K} \right)\right\rfloor}{N}$, the BLER is upper bounded by
\begin{align} \label{eq:thm2}
        {\sf P}({\mathcal E}) \!\leq\! 1 - &\left(1-Q\left(\frac{1}{\sigma}\right)\right)^{2K}\frac{(N\!-\!2K)}{\sqrt{2\pi\sigma^2}}\nonumber\\
        &\int_{0}^{\infty} \!\!\!Q_{\frac{1}{2}}\!\left(\frac{1}{\sigma},\!\frac{\sqrt{y}}{\sigma}\right)^{\!2K}\!\!\left\{1\!-\!2Q\left(\frac{\sqrt{y}}{\sigma}\right)\right\}^{N\!-2K-1}\!\!y^{-\frac{1}{2}}e^{-\frac{y}{\sigma^2}}{\rm d}y,
\end{align}
where $Q_M(a,b)$ denotes the generalized Marcum Q-function of order $M$.

\end{thm} 

\begin{proof}
See Appendix B.
\end{proof}

Although the analytical expressions of the BLERs in Theorem \ref{thm1} and \ref{thm2} are integral forms, they allow us to identify the minimum required SNR to achieve the extremely low BLER performance for a given code rate and code length without heavy Monte Carlo simulations. This is particularly useful for the application of the haptic feedback in tele-surgery, which needs a BLER below $10^{-9}$ \cite{URLLCapp}. We show the exactness of our analytical expressions in the sequel.

\subsection{Achievability of the Ultimate Shannon Limit}


We are now ready to state the main result. 

\begin{thm} 
A single-layered OSS code with $R=\frac{ \log_2\left( {N\choose K}\right)}{N}$ and a simple ordered statistics decoder achieve the ultimate Shannon limit in the power-limited regime.
\end{thm}

\begin{proof} 
See Appendix C.
\end{proof}

The practical implication of the theorem is that the Shannon limit in the power-limited regime is achievable even with a linear complexity decoding algorithm. In fact, biorthogonal codes with maximum-likelihood (ML) decoding also approach the ultimate Shannon limit on $E_b/N_0$ as $N\rightarrow \infty$, albeit with $R\rightarrow 0$. Using the Hadamard transformation, the ML decoding complexity of the biorthogonal codes can be implemented with order $\mathcal{O}(N\log N)$, which is super-linear in the block length \cite{Sloane}. Therefore, our code with a simple decoding scheme can attain unbounded gains in the decoding complexity as the blocklength increases compared to the biorthogonal code with the ML decoding. 

{\bf Remark 6 (Emerging applications):} Our codes can be very useful in emerging power-limited wireless communication scenarios such as terahertz and visible light communications systems, in which the signal bandwidth is enormous, but the propagation loss is very high compared to lower frequency bands, i.e., SNR is limited. For more details, see \cite{tera} and references therein. In these systems, the proposed codes with a large block length can be a practical coded-modulation scheme.

\subsection{Coding Gain Comparison}
   To gauge the code performance in a finite block length, it is instructive to compare the effective coding gain of the proposed OSS code with the biorthogonal code:
  \begin{itemize}
  	\item {\bf Biorthogonal code:}  For any integer $m\in \mathbb{Z}^{+}$, there exists a binary code with $[N,B,d_{\rm min}^H]=[2^{m-1},m,2^{m-2}]$. Therefore, the nominal coding gain of this code is 
  	\begin{align}
  		\gamma_{c}^{{\sf bi}}  =\frac{d^2_{\rm min}(\mathcal{C})/4}{E_s/R} =\frac{\log_2(N)+1}{2}.
  	\end{align}
  	The number of codewords at minimum distance from a given codeword  is $N_{d_{\rm min}}=2^m-2=2N-2$.  The effective coding gain is approximately computed using the rule 0.2-dB loss per factor of 2 for the normalized number of the nearest neighbor per bit $ N_{d_{\rm min}}/\log_2 N =(2N-2)/\log_2 N$ \cite{Forney1998}, which is given by
  		\begin{align}
  		\gamma_{{\sf eff}}^{{\sf bi}}  =	10\log_{10}\left(\frac{\log_2(N)+1}{2}\right)- 0.2\log_2\left(\frac{2N-2}{\log_2 N}\right)~~{\rm dB}.
  	\end{align}


  	
  	\item {\bf OSS code:} Suppose a single-layered orthogonal sparse superposition code with length $N$ and $K_1=1$.  Since $d_{\rm min}^2(\mathcal{C})=2$, $E_s=\frac{1}{N}$, and $R=\frac{\log_2N}{N}$, the nominal coding gain becomes  
  	\begin{align}
 	\gamma_{c}^{{\sf oss}} &=\frac{d^2_{\rm min}(\mathcal{C})/4}{E_s/R}=\frac{\log_2 N}{2}.
 \end{align} 
 	Since the normalized number of the nearest neighbors per bit is $\frac{N-1}{\log_2N}$, the effective coding gain of our code becomes
 	  		\begin{align}
  		\gamma_{{\sf eff}}^{{\sf oss}}  =	10\log_{10}\left(\frac{\log_2 N}{2}\right)- 0.2\log_2\left(\frac{N-1}{\log_2 N}\right)~~{\rm dB}.
  	\end{align}
 	
  \end{itemize}
 From this comparison, we can see that the two codes provide a similar effective coding gain for a finite block length, and the gains of the two codes are asymptotically identical when the block length is infinite. Notwithstanding the similarity, the two codes obtain the effective coding gains in different ways. On the one hand, the biorthogonal code enhances the coding gain by increasing the minimum distance of the codewords, while keeping the energy per bit a constant. On the other hand, our code improves the coding gain by diminishing the energy per bit when transmitting a codeword (i.e., increasing the sparsity of codewords) while maintaining the minimum distance.    
  
 We also compare the nominal and effective coding gains of our code with those of the existing binary linear codes including first-order Reed-Muller (RM) codes and Golay codes in \cite{Forney1998}, which is summarized in Table I. Although the coding gain of our code is slightly less than those of the other codes, all these coding gains are attainable when using the optimal ML decoding. When using sub-optimal (i.e., low-complexity) decoders for the other codes, the gains dwindle.  However, the effective coding gain of our code maintains with the simple linear complexity decoder.

\begin{table}
\caption{Coding gain comparison.}
\label{table:codinggain}
\begin{center}
\begin{tabular}{|c|c|c|c|}
\hline
Code & $\left[N,B, d_{\rm min}^{\sf H}\right] $              & $\gamma_c$ (in dB)   & $\gamma_{\rm eff}$ (in dB)    \\
\hline\hline
Proposed (Two-layered) &$[65,12,2]$ & 4.9   &  3.9  \\
RM & $[64,7,16] $  & 5.4   & 4.4   \\
Golay  &$[64,22,16] $ & 7.4   &  6.0  \\\hline
Proposed (Two-layered) &$[129,14,2] $ &     5.4  & 4.6\\
RM &$[128,8,64] $               & 6.0   &   4.9 \\
Golay &$[128,29,32] $               & 8.6   & 6.9  \\\hline
Proposed (Two-layered) &$[256,16,2] $ & 6.0   &  5.1  \\
RM &$[256,9,128] $              & 6.5   & 5.4   \\
Golay &$[256,37,64] $               & 9.7  & 7.6  \\ 
\hline\hline
\end{tabular} 
\end{center}
\end{table}

 \section{Numerical Results}
In this section, we provide numerical results to demonstrate the effectiveness of the proposed encoding and decoding methods.

\subsection{Validation of BLER Performance Analysis}

To validate our analysis, we compare our analytical expression of the BLER with the numerical results. As can be seen in Fig. \ref{figure:Fig2}, the derived expression perfectly matches to the simulation results in all SNR values, which confirms the exactness of our analysis in Theorem \ref{thm1}. In addition, we observe that our upper bound derived in Theorem \ref{thm2} is very tight for various block lengths as shown in Fig. \ref{figure:Fig3}. In particular, the derived upper bounds are useful to predict the performances of the codes correctly within 0.2 dB SNR loss around $\text{BLER}=10^{-5}$.

\begin{figure}[t]
	\centering 
   \includegraphics[width=0.51\textwidth]{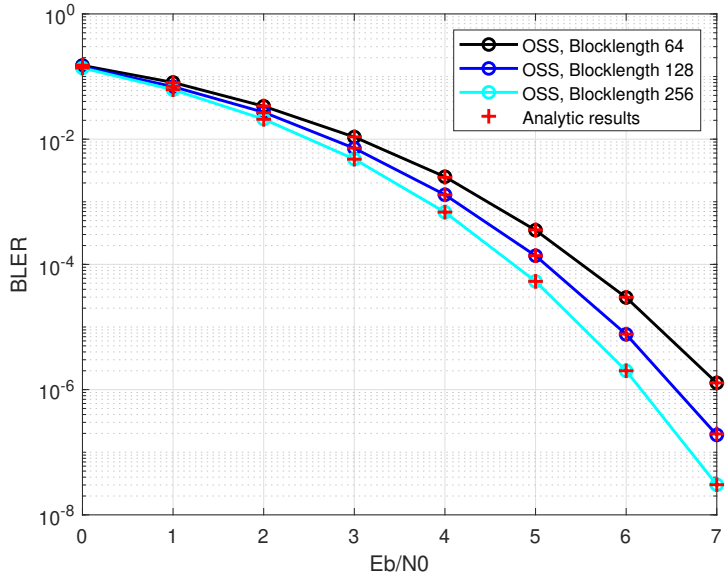}\vspace{-0.1cm}
  \caption{Comparison of the analytical expression of the BLER and the empirical results on BLER for single-layered orthogonal sparse superposition codes.}  \label{figure:Fig2}\vspace{-0.1cm}
\end{figure}

\begin{figure}[t]
	\centering 
   \includegraphics[width=0.51\textwidth]{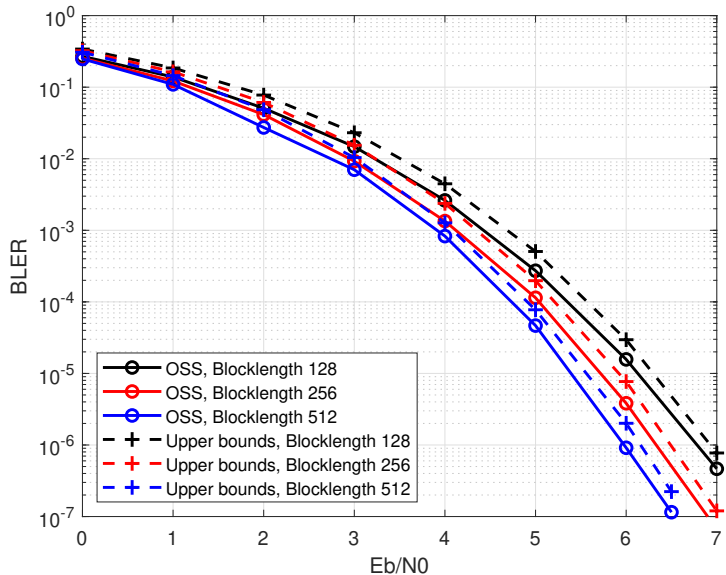}\vspace{-0.1cm}
  \caption{Comparison of the analytical expression of the BLER and the empirical results on BLER for symmetric two-layered orthogonal sparse superposition codes.}  \label{figure:Fig3}\vspace{-0.1cm}
\end{figure}

\subsection{BLER Comparison}
We compare the BLER performance of the OSS codes with that of the existing codes in terms of SNR per bit, i.e., $E_b/N_0$, which is typically used to fairly compare the BLERs for the codes with different rates. We consider a two-layered OSS code with block length $N=257$ and rate $R=\frac{16}{257}$ and adopt the simple ordered statistics decoder, which needs the complexity of $\mathcal{O}\left(N\right)$. For a fair comparison, we evaluate the BLER performances of the low-rate codes under the similar block length $N=256$ and the decoding complexity. The following two conventional low-rates codes are considered:
\begin{itemize}
	\item $[16,32]$-polar codes concatenated with a repetition code with rate $1/8$: this concatenated code provides an effective code rate of $R=\frac{16}{256}$.  We use the generator matrix of the polar code constructed by the  Arikan's kernal matrix and the information set in \cite{Polar}. For decoding, we use a successive cancellation list (SCL) decoder with list size of $L_p=8$. Therefore, the decoding complexity order becomes $\mathcal{O}\left(L_p\times N/8\log(N/8)\right)$. 
	\item $[16,48]$-convolutional codes concatenated with a repetition code with rate $1/4$; the effective code rate is $R=\frac{16}{(48+15)\times 4}$ where $15$ denotes the tail bits of the convolutional code We use the generator polynomial in octal is [54,64,74] with constraint length $K_c=4$. We adopt the soft-Viterbi decoder, which needs the complexity of $\mathcal{O}\left(2^{K_c}\times \log(N/4)\right)$.
\end{itemize} As can be seen in Fig. \ref{figure:Fig5}, the proposed method outperforms the existing low-rate codes under a similar decoding complexity. The conventional low-rate codes with the repetitions cannot attain sufficient coding gains while reducing the decoding complexity. Whereas, our code yields a relatively large coding gain with linear decoding complexity in the block length $N$. Consequently, the proposed code is more promising than the existing codes for URLLC and control channels in LTE systems. 
  


\begin{figure}[t]
	\centering 
   \includegraphics[width=0.53\textwidth]{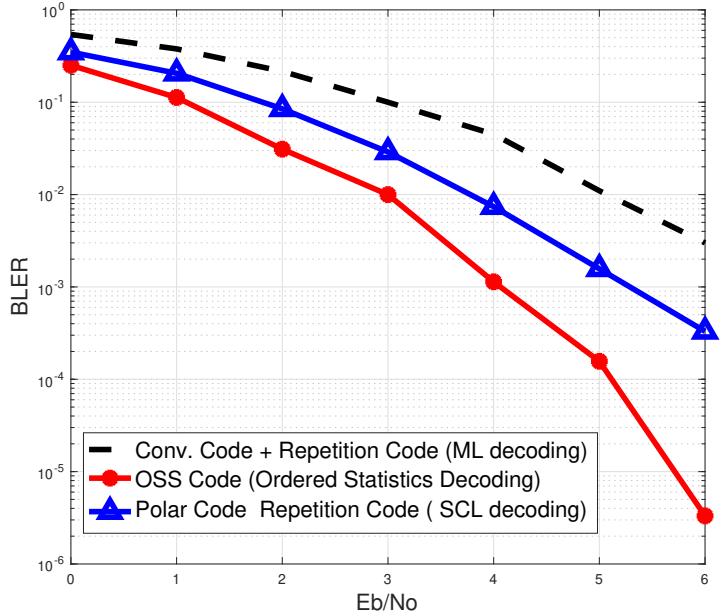}\vspace{-0.1cm}
  \caption{BLER performance comparison of different codes when $N=256$. }  \label{figure:Fig5}\vspace{-0.1cm}
\end{figure}

   
\begin{figure}[t]
	\centering 
   \includegraphics[width=0.53\textwidth]{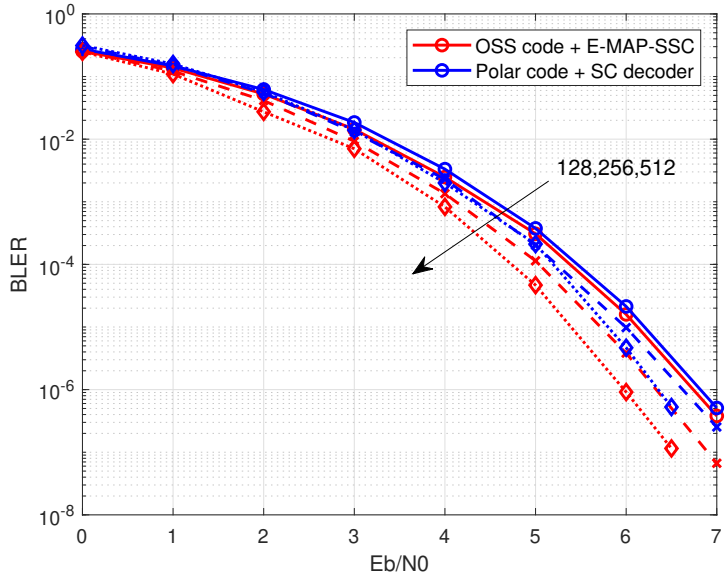}\vspace{-0.1cm}
  \caption{BLER performance comparison when the block length increases while decreasing code rates.}  \label{figure:Fig6}\vspace{-0.1cm}
\end{figure}

We compare the BLER performances between the OSS codes using E-MAP-SSC and the $1/2$-polar codes using SC decoders for different block lengths. As can be seen in Fig. \ref{figure:Fig6}, the proposed encoding and decoding method outperforms the polar codes using the SC decoder when the block lengths are less than five hundred. For example, our method provides approximately 0.4 dB gain over the polar code with the SC decoder at $\text{BLER}=10^{-4}$ when the code rate is $R=\frac{18}{512}$. In addition, with code rate decreasing inversely proportional to the block length, the SNR gap between the OSS and the polar codes increases. Nevertheless, the decoding complexity of the proposed method is much less than that of the SC decoder.

\subsection{Finite Blocklength Performance}

We also evaluate the performances of OSS codes in terms of finite-blocklength achievable rates introduced in \cite{Polyanskiy}. Using finite-block length channel codes, the achievable rates are far below the channel capacity. By allowing a small decoding error probability $\epsilon$, the finite-blocklength achievable rates are defined by 
\[\frac{\log M^{\star}(N,\epsilon)}{N},\]
where $M^{\star}(N,\epsilon)$ is the number of codewords using $N$ channel uses that can be transmitted with averaged BLER up to $\epsilon$. While the exact capacity in a finite blocklength regime is unknown, there exist several converse bounds in \cite{Shannon2}. Further, the finite blocklength achievable rate for length $N$ channel codes can be approximated as
\begin{align}
    \frac{\log M^{\star}(N,\epsilon)}{N}\approx C-\sqrt{\frac{V}{N}}Q^{-1}\left(\epsilon\right)+\frac{\log N}{2N}.
\end{align}
With these benchmarks, we present feasible OSS codes that can provide the highest coding rate for different code lengths. In particular, we compare the feasible coding rates of the OSS codes with E-MAP-SSC decoding to that of the polar codes with successive cancellation (SC) decoding, which needs decoding complexity order of $\mathcal{O}(N\log N)$.

\begin{figure}[t]
	\centering 
   \includegraphics[width=0.52\textwidth]{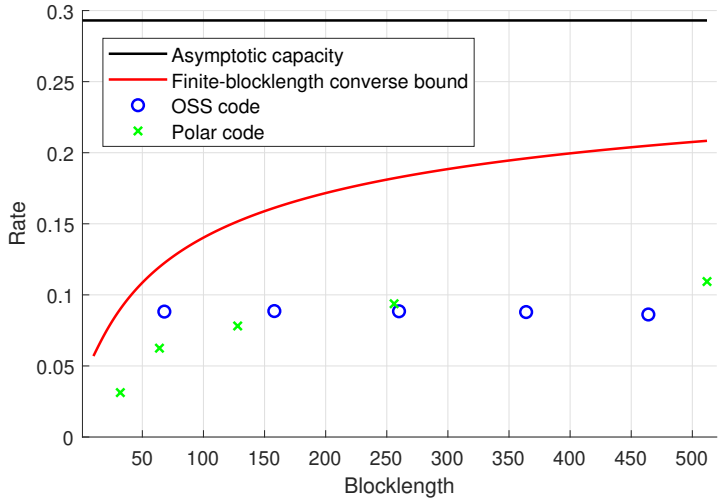}\vspace{-0.1cm}
  \caption{Finite length achievable rate comparison for $\text{SNR}=-3$dB and target error rate of $\epsilon=10^{-3}$.}  \label{figure:Fig4}\vspace{-0.1cm}
\end{figure}

Fig. \ref{figure:Fig4} provides the feasible code rates of different channel codes for a fixed SNR of $-3$dB and the target error tolerance of $10^{-3}$. In the figure, note that the results are slightly biased to the polar codes since the blocklengths of the OSS codes are compulsorily selected to meet the decoding error tolerance of $10^{-3}$. Nevertheless, the OSS codes can provide higher achievable rates than the polar codes even having shorter block lengths for $N<256$. Also, we can observe that the feasible OSS codes are very close to Shannon's converse bound. This result is quite interesting since the required decoding complexity of the E-MAP-SSC is significantly lower than that of the SC decoder, which requires $\mathcal{O}(N\log N)$.

\section{Conclusion}

This paper has introduced a new class of sparse superposition codes, called orthogonal sparse superposition (OSS) codes. To construct this type of codes, we have presented a novel encoding technique to generate codewords that are a sparse linear combination of orthogonal columns of a dictionary matrix. Harnessing the orthogonal structure, we also have proposed a near-optimal decoder with the linear decoding complexity in the block length. We have provided analytical expressions of the BLERs for the proposed codes as a function of code block lengths and rates. One key result is that our code with the linear complexity decoder can achieve the ultimate Shannon limit in the power-limited regime. In comparison with polar codes using the SC decoder and convolutional codes using the ML decoder, we have demonstrated that the proposed encoding and decoding techniques are more effective than the existing coded modulation techniques in the short block length regime.

 	
An important direction for future work would be to study the performances of orthogonal sparse superposition codes for fading channels. Another interesting direction for future work would be to generalize the code construction by relaxing the orthogonality principle.

\section*{Appendix}

\subsection{Proof of Theorem 1}

Before providing the proof, we first introduce a lemma, which provides probability distribution of the maximum and the minimum of $N$ IID random variables.  
\begin{lem} \label{lemma1}We let $\{X_n\}_{n=1}^N$ be a sequence of IID random variables, each with probability density function (PDF) $f_{X}(x)$ and cumulative density function (CDF) $F_{X}(x)$. We also denote the maximum and the minimum of the sequence by $X_{\sf max}=\max\left\{X_1,X_2,\ldots, X_N\right\}$ and $X_{\sf min}=\min\left\{X_1,X_2,\ldots, X_N\right\}$. Then, the distributions of $X_{\sf max}$ and $X_{\sf min}$ are
\begin{align}
    f_{X_{\sf max}}(x) &= N f_{X}(x) F_{X}(x)^{N-1},  \nonumber \\
     f_{X_{\sf min}}(x) &= N f_{X}(x) \left[1 -  F_{X}(x) \right]^{N-1}.
\end{align}
\end{lem}
\begin{proof}
The proof is direct from the IID assumption of $X_k$ for $k\in[N]$.	
\end{proof}

Now, we are ready to prove Theorem \ref{thm1}. Recall that the probability of error does not depend on which codeword is sent thanks to the symmetry of codewords. For ease of exposition, therefore, we focus on computing the probability of error when $	\mathcal{I}=\{1,2,\ldots,K_1\}$. For the AWGN channel, the received signal of the $n$th channel use becomes
\begin{align}
   Y_n = \begin{cases}
    1 +V_n&\text{ for } n\in [K]\\
     V_n           &\text{ for } n\notin [K],
          \end{cases} \nonumber
\end{align}
where $V_n$ is a Gaussian noise with zero-mean and variance $\sigma^2$. As explained in Remark 5, a simple ordered statistics decoder is the optimal decoder. Therefore, the probability of decoding error is computed as 
\begin{align}
	{\sf P}(\mathcal{E})& = 1- {\sf P}\left(\mathcal{E}^c \right)\nonumber\\
	&=1-{\sf P}\left( Y_{\mathcal{I}}^{\rm min} >  Y_{\mathcal{I}^c}^{\rm max} \right),
	\end{align}
where $Y_{\mathcal{I}}^{\rm min}=\min\left\{Y_n\right\}$ for $n\in \mathcal{I}$ and  $Y_{\mathcal{I}}^{\rm max}=\max\left\{Y_n\right\}$ for $n\in \mathcal{I}_c$. Using the conditional expectation theorem, we obtain
\begin{align}
&{\sf P}\left( Y_{\mathcal{I}}^{\rm min} >  Y_{\mathcal{I}^c}^{\rm max} \right) \nonumber\\
    &=  \mathbb{E}_{Y^{\sf {max}}_{\mathcal{I}^{c}}}\left[{\sf P}\left(Y_{\mathcal{I}}^{\rm min} > y\right)\mid   Y_{\mathcal{I}^c}^{\rm max}=y \right]\nonumber \\
              &=\int_{-\infty}^{\infty} {\sf P}\left(Y_{\mathcal{I}}^{\rm min} > y\right) f_{Y^{\sf {max}}_{\mathcal{I}^{c}}}(y) {\rm d}y, \nonumber \\
          &=  \int_{-\infty}^{\infty} {\sf P}\left(Y_{\mathcal{I}}^{\rm min} > y\right)(N-K_1)f_{Y_{\mathcal{I}^{c}}}(y)F_{Y_{\mathcal{I}^{c}}}(y)^{N-K_1-1} {\rm d}y, \nonumber \\          
          &=  \int_{-\infty}^{\infty} \!\! \left[1 - F_{Y_{\mathcal{I}}}\!\left(y\right) \right]^{K}\!\!  \left(N-K_1 \right) f_{Y_{\mathcal{I}^{c}}}(y)F_{Y_{\mathcal{I}^{c}}}(y)^{N-K_1-1} {\rm d}y. \label{eq:cor_prob_supp}
\end{align}
Since $Y_n$ is distributed as $\mathcal{N}(1,\sigma^2)$ for $n\in[K]$ and $\mathcal{N}(0,\sigma^2)$ for $n\notin[K]$. Therefore, invoking $F_{Y_{\mathcal{I}}}(y)=1-Q\left(\frac{y-1}{\sigma}\right)$ and $F_{Y_{\mathcal{I}^{c}}}(y)=1-Q\left(\frac{y}{\sigma}\right)$ into \eqref{eq:cor_prob_supp}, we obtain the expression in \eqref{eq:thm1}.

\subsection{Proof of Theorem 2}
To prove this, we consider a sub-optimal decoder with a two-stage decoding algorithm:
\begin{itemize}
    \item Stage 1: The decoder first identifies the $2K$ non-zero support elements that produce the largest magnitude values in ${\bf |Y|}$.
    \item Stage 2: The decoder assigns $1$ for the selected indices if its sign is positive. Otherwise, it allocates $-1$. \end{itemize}
Let ${\sf P}\left(\hat{{\mathcal E}}_{1}^{c}\right)$ and ${\sf P}\left(\hat{{\mathcal E}}_{2}^{c}\right)$ be the probability of the success in decoding at stage 1 and stage 2, respectively. By the above decoding rule, the probability of decoding the support set correctly at stage 1 is \begin{align}
    {\sf P}\left(\hat{{\mathcal E}}_{1}^{c}\right)={\sf P}\left( |Y_{\mathcal{I}}^{\rm min}|^2 >  |Y_{\mathcal{I}^c}^{\rm max}|^2 \right). \label{error1}
\end{align}
From Lemma \ref{lemma1}, we obtain
\begin{align}
&{\sf P}\left( |Y_{\mathcal{I}}^{\rm min}|^2 >  |Y_{\mathcal{I}^c}^{\rm max}|^2 \right) \nonumber\\
    &=  \!\mathbb{E}_{|Y^{\sf {max}}_{\mathcal{I}^{c}}|^2}\left[{\sf P}\left(|Y_{\mathcal{I}}^{\rm min}|^2 > y\right)\mid   |Y_{\mathcal{I}^c}^{\rm max}|^2=y \right]\nonumber \\
              &=\!\int_{0}^{\infty} \!\!\!{\sf P}\left(|Y_{\mathcal{I}}^{\rm min}|^2 \!> \!y\right) f_{|Y^{\sf {max}}_{\mathcal{I}^{c}}|^2}(y) {\rm d}y \nonumber \\
          &= \! \int_{0}^{\infty} \!\!\!{\sf P}\left(|Y_{\mathcal{I}}^{\rm min}|^2\! > \!y\right)\!(N\!-\!2K)f_{|Y_{\mathcal{I}^{c}}\!|^2}(y) F_{|Y_{\mathcal{I}^{c}}\!|^2}(y)^{N-2K-1} {\rm d}y \nonumber \\          
          &= \left(N\!-\!2K \right)\!\! \int_{0}^{\infty} \!\!\! \left[1 \!-\! F_{|Y_{\mathcal{I}}\!|^2}\!\left(y\right) \right]^{2K}\!\!  \! f_{|Y_{\mathcal{I}^{c}}\!|^2}(y) F_{|Y_{\mathcal{I}^{c}}\!|^2}(y)^{N-2K-1} {\rm d}y. \label{eq:cor_prob_supp}
\end{align}
Note that $|Y_{n}|^2$ for $n\in \mathcal{I}$ follows a scaled non-central chi-square distribution with one degree of freedom and the non-centrality parameter of $\frac{1}{\sigma^2}$, i.e., $\sigma^2{\chi}^2\left(1,\frac{1}{\sigma^2}\right)$. In addition, the distribution of $|Y_{n}|^2$ for $n\notin \mathcal{I}$ is the scaled chi-square with one degree of freedom, i.e., $\sigma^2\chi^2\left(1\right)$. As a result, the probability that correctly identifies the non-zero support set is 
\begin{align}
    &{\sf P}\left(\hat{{\mathcal E}}_{1}^{c}\right)=\nonumber\\&\frac{(N\!-\!2K)}{\sqrt{2\pi\sigma^2}}\!\!\int_{0}^{\infty} \!\!\!Q_{\frac{1}{2}}\!\left(\frac{1}{\sigma},\!\frac{\sqrt{y}}{\sigma}\right)^{\!2K}\!\!\left\{1\!-\!2Q\left(\frac{\sqrt{y}}{\sigma}\right)\right\}^{N\!-2K-1}\!\!y^{-\frac{1}{2}}e^{-\frac{y}{2\sigma^2}}{\rm d}y.
\end{align}
Under the premise that $\hat{\mathcal{I}}=\left\{\mathcal{I}_1\cup\mathcal{I}_2\right\}$, decoding is successful if the signs of $Y_{n}$ for $n\in \mathcal{I}$ are not flipped.  As a result, the probability that the signs of the $2K$ non-zero elements are not changed is  
\begin{align} 
    {\sf P}\left(\hat{{\mathcal E}}_{2}^{c}\right)=\left(1-Q\left(\frac{1}{\sigma}\right)\right)^{2K}.\label{error2}
\end{align}
Using \eqref{error1} and \eqref{error2}, the BLER of the sub-optimal decoder becomes
\begin{align}
    {\sf P}(\hat{{\mathcal E}})&=1-{\sf P}(\hat{{\mathcal E}}^{c}) \nonumber\\
    &=1-{\sf P}\left(\hat{{\mathcal E}}_{1}^{c}\right){\sf P}\left(\hat{{\mathcal E}}_{2}^{c}\right).
\end{align}
Since the BLER of the optimal decoder is upper bounded by that of this sub-optimal decoder, we completes the proof.

\subsection{Proof of Theorem 3}
In this proof, we aim to show that, with the simple our code and the ordered statistics decoder, the probability of  decoding error approaches zero as $N\rightarrow \infty$ and $R=\frac{ \log_2\left( {N\choose K}\right)}{N}\rightarrow 0$, provided that $E_b/N_o > \ln 2$ (1.53 dB).

 We commence by introducing the following  lemma, which is instrumental for proving our main theorem.
\begin{lem}\label{lem2}
For any $K\geq 1$ and $\beta>1$, 
\begin{align}
	\lim_{N \rightarrow \infty} \exp\left(-\frac{\beta}{K} \ln {N\choose K} +\ln(N-K)\right)  \rightarrow 0.
\end{align}
\end{lem}

\begin{proof}
To prove this lemma, we use the following sandwich inequality:
\begin{align}
	\frac{(N-K)^K}{K!} \leq {N\choose K}  \leq \frac{N^K}{K!}.
\end{align}
By taking the exponent $-\frac{\beta}{K}$ and multiplying $N-K$, we obtain
\begin{align}
	(N-K) \frac{\left(K! \right)^{\frac{\beta}{K}}}{N^{\beta}} \leq  (N-K) {N\choose K}^{-\frac{\beta}{K}}  \leq  (N-K) \frac{\left(K! \right)^{\frac{\beta}{K}}}{(N-K)^{\beta}}.
\end{align}
For any $K\geq 1$ and $\beta>1$, both the upper and lower bound approach zero as $N$ goes to infinity, i.e., $\lim_{N\rightarrow \infty}(N-K) \frac{\left(K! \right)^{\frac{\beta}{K}}}{N^{\beta}}=0$ and $\lim_{N\rightarrow \infty} \frac{\left(K! \right)^{\frac{\beta}{K}}}{N^{\beta-1}}=0$. Therefore, by the squeeze theorem, we conclude that 
\begin{align}
	\lim_{N\rightarrow \infty}(N-K) {N\choose K}^{-\frac{\beta}{K}}\rightarrow 0,
\end{align}
which competes the proof. 
\end{proof}

  The probability error does not depend on which codeword was transmitted by the symmetry of orthogonal sparse superposition codes. Thus, we focus on an orthogonal sparse superposition codeword with length $N$ and nonzero support set $\mathcal{I}=\{1,2,\ldots, K\}$.  Using a simple ordered statistics decoder with linear decoding complexity, the codeword is correctly decodable if $\min_{i\in \mathcal{I}}\left\{Y_{i}\right\} >\max_{j\in [N]/\mathcal{I}}\left\{Y_{j}\right\}$. We define such event set as
\begin{align}
	\mathcal{E}^c=\left\{\min_{i\in \mathcal{I}}\left\{Y_{i}\right\} >\max_{j\in [N]/\mathcal{I}}\left\{Y_{j}\right\}\right\}.
\end{align}
Then, the decoding error probability is  \begin{align}
	{\sf P}(\mathcal{E})& = 1- {\sf P}\left(\mathcal{E}^c \right).
	\end{align}
To compute a lower bound for the probability of the correct decoding, we introduce a subset of $\mathcal{E}^c$ such that
\begin{align}
	\mathcal{E}^c_{\delta}=\left\{\min_{i\in \mathcal{I}}\left\{Y_{i}\right\} > 1-\delta\right\} \cap \left\{1-\delta > \max_{j\in [N]/\mathcal{I}}\left\{Y_{j}\right\}\right\}
\end{align}
for $0<\delta<1-\sqrt{\frac{2}{2+\epsilon}}$, where $\epsilon$ is an arbitrarily small positive constant. Using this subset, the lower bound of the correct decoding probability is computed as
\begin{align}
		{\sf P}\left({\mathcal E}^c \right) &\geq {\sf P}\left({\mathcal E}^c_{\delta}\right)\nonumber\\
	&= {\sf P}\left( \min_{i_k\in \mathcal{I}}\left\{Y_{i_k}\right\} > 1-\delta   \right){\sf P}\left(  1-\delta  > \max_{n\in [N]/\mathcal{I}}\left\{Y_{n}\right\} \right)\nonumber\\
	&\stackrel{(a)}{=} \left\{ {\sf P}\left(  Y_{i_k}  > 1-\delta   \right) \right\}^K \left\{ {\sf P}\left(   1-\delta  >  Y_n \right)\right\}^{N-K}\nonumber\\
	&\stackrel{(b)}{=} \left\{1-Q\left(\frac{\delta}{\sigma}\right)\right\}^{K} \left\{1-Q\left(\frac{1-\delta}{\sigma}\right)\right\}^{N-K}
\nonumber\\
&\stackrel{(c)}{\geq} \left\{1-KQ\left(\frac{\delta}{\sigma}\right)\right\} \left\{1-(N-K)Q\left(\frac{1-\delta}{\sigma}\right)\right\},\label{eq:low_cor_prob}
\end{align}
where (a) follows from the independence of $Y_n$ for $n\in [N]$ and (b) follows since $Y_{n}$ for $n\in [N]$ are Gaussian random variables. In addition, (c) follows from the Bernoulli's inequality, i.e., $(1-x)^r\geq 1-rx$ for every integer $r \geq 0$ and $x \geq 0$. Applying the Chernoff bound. i.e., $Q(x)\leq \exp\left(-\frac{x^2}{2}\right)$ for $x\geq 0$ in \eqref{eq:low_cor_prob}, we establish an upper bound for the probability of decoding error in a closed form for given code block length $N$, the sparsity level of the codeword $K$, and noise variance $\sigma^2$ as
\begin{align}
		{\sf P}\left({\mathcal E} \right) &\leq 1- \left\{1\!-\!K\exp\!\left(\!-\frac{\delta^2}{2\sigma^2}\!\right)\right\} \left\{1\!-\!(N\!-\!K)\!\exp\!\left(\!-\frac{(1\!-\!\delta)^2}{2\sigma^2}\!\right)\!\right\} \nonumber \\
		&=K\exp\left(-\frac{\delta^2}{2\sigma^2}\right) 	\!+\!(N\!-\!K)\exp\left(-\frac{(1\!-\!\delta)^2}{2\sigma^2}\!\right) \nonumber\\
		& ~~~- K(N-K)\exp\left(-\frac{\delta^2 +(1-\delta)^2}{2\sigma^2}\right).\label{eq:error_upper}
		\end{align} 
Now, we express this upper bound as a function of $E_b/N_o$. From the definition, i.e.,
\begin{align}
	\frac{E_{\rm b}}{N_{\rm o}} = \frac{1}{2R}\text{SNR},
\end{align}
we have
\begin{align}
	\frac{1}{\sigma^2} = \frac{2}{K}\log_2 {N\choose K}\frac{E_{\rm b}}{N_{\rm o}}, \label{eq:ebno_sigma}
	\end{align}
since the SNR is $\frac{K}{N\sigma^2}$ and the code rate is $R=\frac{\log_2\left({N\choose K}\right)}{N}$. Invoking \eqref{eq:ebno_sigma} into \eqref{eq:error_upper}, the decoding error probability is upper bounded by 
\begin{align}
	{\sf P}(\mathcal{E})  &\leq\exp\left(-\frac{\delta^2}{K\ln 2}\ln{N\choose K}\frac{E_{\rm b}}{N_{\rm o}} + \ln K\right) \nonumber\\
	& + \exp\left(-\frac{(1-\delta)^2}{K\ln(2)}\ln{N\choose K}\frac{E_{\rm b}}{N_{\rm o}} + \ln (N-K)\right) \nonumber\\
		& - \exp\left( -\frac{\delta^2 +(1-\delta)^2}{K\ln 2}\ln{N\choose K}\frac{E_{\rm b}}{N_{\rm o}} +\ln(K)+\ln (N-K)\right).\label{eq:error_upper2}
\end{align}
Notice that the first term in \eqref{eq:error_upper2} approaches zero as $N$ goes to infinity. Since the third term in \eqref{eq:error_upper2} is the multiplication of the first and the second term, to claim ${\sf P}(\mathcal{E})\rightarrow 0$, it is sufficient to show that the second in \eqref{eq:error_upper2} becomes zero as $N$ goes to infinity for any $\frac{E_{\rm b}}{N_{\rm o}}$ greater than $\ln 2$. For an arbitrary $\frac{E_{\rm b}}{N_{\rm o}}$ greater than $\ln 2$, we consider a sufficiently small positive value $\epsilon>0$ that satisfies
\begin{align}
	\frac{E_{\rm b}}{N_{\rm o}} =\ln2 +\frac{\ln2}{2}\epsilon.
\end{align}
By plugging this, we can rewrite the second term in \eqref{eq:error_upper2} as
\begin{align}
    \exp\left(-\frac{(1-\delta)^2}{K}\left(1+\frac{\epsilon}{2}\right)\ln{N\choose K} + \ln (N-K)\right).
\end{align}
From Lemma \ref{lem2}, this term approaches zero, provided that
\begin{align}
    (1-\delta)^2(1+\frac{\epsilon}{2})>1,
\end{align}
which is true due to the choice of $\delta$ as $0<\delta<1-\sqrt{\frac{2}{2+\epsilon}}$. Consequently, we conclude that
\begin{align}
	\lim_{N \rightarrow \infty} P(\mathcal{E}) \rightarrow 0,
\end{align}
for any $\frac{E_{\rm b}}{N_{\rm o}}>\ln2$, which completes the proof.


\end{document}